\documentclass[twoside,11pt]{article}

\usepackage{graphicx, subfigure}
\usepackage[top=1in,bottom=1in,left=1in,right=1in]{geometry}
\usepackage{cite}

\usepackage{amsfonts, amsmath, amssymb, amsthm, constants}

\usepackage{color}
\definecolor{darkred}{RGB}{150,0,0}
\definecolor{darkgreen}{RGB}{0,150,0}
\definecolor{darkblue}{RGB}{0,0,200}

\usepackage{hyperref}
\hypersetup{colorlinks=true, linkcolor=darkred, citecolor=darkgreen, urlcolor=darkblue}
\usepackage{url}

\newtheorem{prp}{Proposition}

\def\beq{\begin{equation}}
\def\eeq{\end{equation}}
\def\beqn{\begin{eqnarray*}}
\def\eeqn{\end{eqnarray*}}
\def\bitem{\begin{itemize}}
\def\eitem{\end{itemize}}
\def\benum{\begin{enumerate}}
\def\eenum{\end{enumerate}}
\def\bmult{\begin{multline*}}
\def\emult{\end{multline*}}

\newcommand{\prpref}[1]{Proposition~\ref{prp:#1}}

\def\bA{\mathbf{A}}
\def\bB{\mathbf{B}}

\def\bI{\mathbf{I}}

\def\bQ{\mathbf{Q}}

\def\bh{\mathbf{h}}

\def\bu{\mathbf{u}}
\def\bv{\mathbf{v}}
\def\bw{\mathbf{w}}
\def\bx{\mathbf{x}}
\def\by{\mathbf{y}}
\def\bz{\mathbf{z}}

\def\bbR{\mathbb{R}}

\begin{document}

\title{Noise Folding in Compressed Sensing}

\author{Ery~Arias-Castro\thanks{E.~Arias-Castro is with the Department of Mathematics, University of California, San Diego, CA 92093 USA ({\tt eariasca@ucsd.edu}).  His work is partially supported by the Office of Naval Research under Grant N00014-09-1-0258.} \ 
and Yonina~C.~Eldar\thanks{Y.~C.~Eldar is with the Department of Electrical Engineering, TechnionÑ
Israel Institute of Technology, Haifa 32000, Israel ({\tt yonina@ee.technion.ac.il}).  Her work is supported in part by the Israel Science Foundation under Grant 1081/07, and by the Ollendorf foundation.}}% <-this % stops a space

%\markboth{IEEE SIGNAL PROCESSING LETTERS}%
%{Arias-Castro and Eldar: Compressed Sensing when the Signal is Noisy}

\maketitle

\begin{abstract}
%\boldmath
The literature on compressed sensing has focused almost entirely on settings where the signal is noiseless and the measurements are contaminated by noise.  In practice, however, the signal itself is often subject to random noise prior to measurement.  We briefly study this setting and show that, for the vast majority of measurement schemes employed in compressed sensing, the two models are equivalent with the important difference that the signal-to-noise ratio is divided by a factor proportional to $p/n$, where $p$ is the dimension of the signal and $n$ is the number of observations. 
Since $p/n$ is often large, this leads to noise folding which can have a severe impact on the SNR.

\medskip

\noindent {\bf Keywords:} 
Compressed sensing, matching pursuit, sparse signals, analog noise vs.~digital noise, noise folding.
\end{abstract}


%\IEEEpeerreviewmaketitle

\section{Introduction}

The field of compressed sensing (CS), focused on recovery of sparse vectors from few measurements, has been attracting vast interest in recent years due to its potential use in numerous signal processing applications \cite{donoho-CS,DuartE_Structured,CRT2}. 
The standard CS setup assumes that we are given measurements
\beq \label{w}
\by = \bA \bx + \bw,
\eeq
where $\by \in \bbR^n$ is the measurement vector, $\bA \in \bbR^{n \times p}$ is the measurement matrix with $n \ll p$, and $\bw \in \bbR^n$ is additive noise.  The signal
$\bx \in \bbR^p$ is assumed to be $s$-sparse, so that no more than $s$ elements of $\bx$ are nonzero. It is also common to assume that $\bx$ is deterministic, an assumption we make throughout.
To recover $\bx$ from $\by$ a variety of algorithms have been developed. These include greedy
algorithms, such as thresholding and orthogonal matching
pursuit (OMP) \cite{PatiRK_Orthogonal}, and relaxation methods, such as basis pursuit \cite{donoho-BP}
(also known as the Lasso) and the Dantzig selector~\cite{CandeT_Dantzig}. 

An important aspect of CS analysis is to develop bounds on the recovery performance of these methods in the presence of noise.
Two standard approaches to modeling the noise $\bw$
is either to assume that $\bw$
is deterministic and bounded \cite{CRT2}, or that $\bw$ is a white noise vector, typically Gaussian \cite{CandeT_Dantzig,BRT09,Ben-Haim:2010:CPG:1878287.1878292,BE10}. The former setting leads to a worst-case
analysis in which an estimator must perform adequately
even when the noise maximally damages the measurements.
By contrast, if
one assumes that the noise is random, then better performance bounds can be obtained. We therefore focus here on the random noise scenario.

Two standard measures used to analyze the behavior of CS recovery techniques are the coherence and the restricted isometry property (RIP)~\cite{intro-CS}.
% {\bf need references here for these} 
If the coherence and RIP of the measurement matrix are small enough, then standard recovery methods such as basis pursuit, OMP and threhsolding can recover $\bx$ from $\by$ with squared-error that is proportional to the sparsity level $s$ and the noise variance $\sigma^2$, times a factor that is logarithmic in the signal length $p$~\cite{CandeT_Dantzig,BRT09,Ben-Haim:2010:CPG:1878287.1878292,BE10}.

In many practical settings, noise is introduced to the signal $\bx$ prior to measurement. As an example, one of the applications of CS is to the design of sub-Nyquist A/D converters. In this context $\bx$ represents the analog signal at the entrance to the A/D converter, which is typically contaminated by noise~\cite{xampling}.
Though important in practice, the prolific literature on CS has not treated signal noise in detail. Recently, several papers raised this important issue \cite{treichler-davenport,BME10,5470150}.
These works all point out the fact that noise present in $\bx$ can have a severe impact on the recovery performance. Here we analyze this setting in more detail in the context of the finite-CS system (\ref{w}), in contrast to the analog setting treated in \cite{BME10}, and provide theoretical justification to these previous observations.
In particular, we show that under appropriate conditions on the measurement matrix $\bA$, the effect of pre-measurement noise is to degrade the signal-to-noise ratio (SNR) by a factor of $n/p$. In systems in which $p \ll n$ this degradation may be severe.

\section{Noise Folding}

\subsection{Problem Formulation}

The basic CS model (\ref{w}) is adequate when the noise is introduced at the measurement stage, so that $\bw$ represents the measurement error or noise. However, in many practical scenarios, noise is added to the signal $\bx$ to be measured,  which is not accounted for in (\ref{w}). A more appropriate description in this case is
\begin{equation}
\label{pren}
\by = \bA (\bx + \bz) + \bw,
\end{equation}
where $\bz$ represents the signal noise, i.e., additive noise that is part of the signal being measured.  Our goal in this letter is to analyze the effect of pre-measurement noise $\bz$ on the behavior of CS recovery methods. We assume throughout that $\bw$ is a random noise vector with covariance $\sigma^2 \bI$, and that similarly $\bz$ is a random noise vector with covariance $\sigma_0^2 \bI$, independent of $\bw$.
Under these assumptions, we show that
(\ref{pren}) is equivalent to 
\begin{equation}
\label{prenm}
\by = \bB \bx +  \bu,
\end{equation}
where $\bB$ is a matrix whose coherence and RIP constants are very close to those of $\bA$, and $\bu$ is white noise with variance $\sigma^2 + (p/n) \sigma_0^2$.

It follows that in order to study (\ref{pren}) we can apply the results developed for (\ref{w}), with the important difference that the noise associated with (\ref{pren}) is larger by a factor proportional to $p/n$. When $n \ll p$ this leads to a large noise increase, or {\em noise folding}. This effect is a result of the fact that the measurement matrix $\bA$ aliases, or combines, all the noise elements in $\bz$, even those corresponding to zero elements in $\bx$, thus leading to a large noise increase in the compressed measurements.

\subsection{Equivalent Formulation}

To establish our results, we note that 
\eqref{pren} can be written as
\beq \label{w+z}
\by = \bA \bx + \bv, 
\eeq
with $\bv$ defined by
\beq \label{v}
\bv = \bw + \bA \bz.
\eeq
Under the assumption of white noise, the effective noise vector $\bv$ has covariance $\bQ$, where
\beq \label{cov}
\bQ = \sigma^2 \bI + \sigma_0^2 \bA \bA^T.
\eeq
As can be seen, in general $\bv$ is no longer white, which complicates the recovery analysis.

A simple special case is when $\bA\bA^T$ is proportional to the identity, so that $\bv$ is still white noise. As an example, suppose that $\bA$ is the concatenation of $r = p/n$ orthonormal bases, i.e., $\bA = [\bA_1 \cdots \bA_r]$, where each $\bA_k$ is an $n \times n$ orthogonal matrix---for example, we may want to analyze a signal with a few bases (e.g, wavelets and sinusoids) as in~\cite{donoho-BP}.  In this case,
\[
\bA \bA^T = \bA_1 \bA_1^T + \cdots + \bA_r \bA_r^T = r \bI=\frac{p}{n} \bI
\]
so that the noise covariance \eqref{cov} becomes $\bQ = \gamma \bI$ with
\beq \label{cov-ortho}
\gamma = \sigma^2 + \frac{p}{n} \sigma_0^2.
\eeq
In this special case the models of (\ref{w+z}) (or (\ref{pren})) and (\ref{w}) are identical, with the only difference being that the noise variance of $\bv$ has increased by $\gamma/\sigma^2$ with respect to that of $\bw$. Assuming that $\sigma_0 \approx \sigma^2$ the increase in noise is proportional to $p/n$, a simple case of {\em noise folding}.

In the next section we show that this result holds more generally. Namely, the models of (\ref{w+z}) and (\ref{w}) are roughly equivalent with a noise increase of $\gamma/\sigma^2$ even when $\bA\bA^T$ is not proportional to the identity.

\section{RIP and Coherence with Whitening}

Consider now a more general CS setting, where $\bA$ is an arbitrary matrix with low coherence or low RIP. To proceed, we first whiten the noise $\bv$ in (\ref{w+z})  by multiplying the linear system by $\bQ_1^{-1/2}$, where $\bQ_1 := \bQ/\gamma$, obtaining the equivalent system
\begin{equation}
\label{prenm2}
\by = \bB \bx +  \bu, \quad \bB := \bQ_1^{-1/2} \bA,\bu := \bQ_1^{-1/2} \bv.
\end{equation}
Now, the noise vector $\bu$ is white with covariance matrix $\gamma \bI$, just as in the case of $\bA\bA^T$ proportional to the identity. The main difference, however, is that the whitening changed the measurement matrix from $\bA$ to $\bB$.  We quantify the magnitude of these changes below via the RIP constants and the coherence.
As we show, for standard matrices used in CS, the change is generally not very significant.

Our results hinge on approximating $\bA \bA^T$ by $(p/n) \bI$ even when $\bA$ is arbitrary.  Let
\beq \label{eta}
\eta := \|\bI - (n/p) \bA\bA^T\|_2,
\eeq
measure the quality of this approximation, where $\|\cdot\|_2$ denotes the standard operator norm in $\bbR^n$. In our derivations below we will assume that $\eta$ is small. Under this assumption we will show that the coherence and RIP constants of $\bB$ and $\bA$ are very similar.
 
To justify the assumption that $\eta$ is often small, note that
when the entries of $\bA$ are i.i.d.~zero-mean, variance $1/n$ random variables with a sub-gaussian distribution (including the normal, uniform, Bernoulli distributions), or when the column vectors are i.i.d.~uniform on the sphere, then $\eta \leq C \sqrt{n/p}$ with probability at least $1 -\exp(-c n)$, for  constants $C, c > 0$ depending only on the distribution of $\bA$~\cite[Th.~39]{vershynin}.  
%In the case of measurement matrix with i.i.d.~Gaussian entries, where the exact order of magnitude for $\eta$ is known.  Indeed, when 
For example, when the entries of $\bA$ are i.i.d.~$N(0,1/n)$ and $n$ is large enough, 
\beq \label{eta-ub}
\eta \leq 2 \sqrt{n/p} + n/p + 4 t/\sqrt{p},
\eeq
with probability at least $1 -2\exp(-t^2/2)$ if $0 < t \leq \sqrt{n}$ and $p \geq n$~\cite[Cor.~35]{vershynin}.
A similar result holds for other distributions, including heavy-tailed distributions, the basic requirement being that the column vectors of $\bA$ are independent with covariance $\bI/n$~\cite{vershynin}. These assumptions are standard in the CS literature. Thus, in these prevalent settings, $\eta$ will be small with high probability.
%
%For illustration purposes, we will use the case of measurement matrix with i.i.d.~Gaussian entries, where the exact order of magnitude for $\eta$ is known.  Indeed, when the entries of $\bA$ are i.i.d.~$N(0,1/n)$, $\eta \leq 2 \sqrt{n/p} + n/p + 4 t/\sqrt{p}$ with probability at least $1 -2\exp(-t^2/2)$ if $0 < t \leq \sqrt{n}$ and $p \geq n$, with $n$ large enough~\cite[Cor.~35]{vershynin}. 

%\begin{thm}[Corollary 35 in \cite{vershynin}]
%When the entries of $\bA$ are i.i.d.~$N(0,1/n)$, $\eta \leq 2 \sqrt{n/p} + n/p + 4 t/\sqrt{p}$ with probability at least $1 -2\exp(-t^2/2)$ if $0 < t \leq \sqrt{n}$ and $p \geq n$, with $n$ large enough.  
%\end{thm}

\subsection{Restricted Isometry Analysis}

We begin by showing that the restricted isometry constants of $\bB$ and $\bA$ are similar assuming a small value of $\eta$.

For an index set $\Lambda \subset \{1, \dots, p\}$ of size $s$, let $\bA_\Lambda$ denote the submatrix of $\bA$ made of the column vectors indexed by $\Lambda$. We say that $\bA$ has the RIP with constants $0 <\alpha_s \leq \beta_s$ if
\beq \label{RIP}
\alpha_s \|\bh \|^2 \leq \|\bA_\Lambda \bh \|^2 \leq \beta_s \|\bh \|^2, \ \forall \bh \in \bbR^s,
\eeq
for any index set $\Lambda \subset \{1, \dots, p\}$ of size~$s$.
The following proposition relates the RIP constants of $\bB$ and $\bA$: 
\begin{prp} \label{prp:RIP}
Assume that $\eta < 1/2$ in \eqref{eta} and that $\bA$ satisfies the RIP of order $s$ with constants $0 <\alpha_s \leq \beta_s$.
Then $\bB$ satisfies the RIP of order $s$ with constants $\alpha_s (1-\eta_1)$ and $\beta_s (1+\eta_1)$, where $\eta_1 := \eta/(1-\eta)$.
\end{prp}
Though the bound is valid for $\eta < 1$, the smaller RIP constant for $\bB$ is only positive when $\eta < 1/2$, leading to our restriction.  
%Concretely, when the entries of $\bA$ are i.i.d.~$N(0,1/n)$, $\eta \sim 2 \sqrt{n/p} + n/p$ (in probability), so that $\eta < 1/2$ for $n/p < 0.05$, when $n$ and $p$ are both large.} 
\begin{proof}
The proof is based on the fact that $\bQ_1$ is close to $\bI$.  Indeed,
by definition of $\eta$ in \eqref{eta},
\beqn
\|\bQ_1 - \bI\|_2 &=& \frac{\sigma_0^2}\gamma \|\bA \bA^T -(p/n)\bI\|_2 \\
&=& \frac{\sigma_0^2 (p/n)}{\sigma^2 + \sigma_0^2 (p/n)} \eta \leq \eta.
\eeqn
Next, we express $\bQ_1^{-1} -\bI$ as a power series
$$
\bQ_1^{-1} -\bI = (\bI - (\bI -\bQ_1))^{-1} -\bI = \sum_{k \geq 1} (\bI -\bQ_1)^k,
$$
which converges since $\|\bQ_1^{-1} - \bI\|_2 \leq \eta < 1$ and $\|\cdot\|_2$ is an operator norm.  Taking norms on both sides of the inequality and using both the triangle inequality and again the fact that $\|\cdot\|_2$ is an operator norm, we obtain
\begin{eqnarray}
\|\bQ_1^{-1} - \bI\|_2
&\leq& \sum_{k \geq 1} \|\bQ_1 -\bI\|_2^k \nonumber \\ 
&\leq& \sum_{k \geq 1} \eta^k = \frac{\eta}{1 - \eta} = \eta_1. \label{eta1}
\end{eqnarray}

Let $\Lambda$ be an index set of size $s$ and take any $\bh \in \bbR^s$.  Then,
$$
\|\bB_\Lambda \bh\|^2 - \|\bA_\Lambda \bh\|^2 = \bh^T \bA_\Lambda^T (\bQ_1^{-1} -\bI) \bA_\Lambda \bh.
$$
Since
\beqn
\left| \bh^T \bA_\Lambda^T (\bQ_1^{-1} -\bI) \bA_\Lambda \bh \right|
&\leq& \|\bQ_1^{-1} - \bI\|_2 \, \|\bA_\Lambda \bh\|^2 \\
&\leq& \eta_1  \, \|\bA_\Lambda \bh\|^2,
\eeqn
we have that
$$
(1 -\eta_1) \|\bA_\Lambda \bh\|^2 \leq \|\bB_\Lambda \bh \|^2 \leq (1 +\eta_1) \|\bA_\Lambda \bh\|^2.
$$
Together with \eqref{RIP}, we obtain
$$
\alpha_s (1 -\eta_1) \|\bh \|^2 \leq \|\bB_\Lambda \bh \|^2 \leq \beta_s (1+\eta_1) \|\bh \|^2,
$$
which concludes the proof.
\end{proof}

\subsection{Coherence Analysis}

We now turn to analyze the coherence. Denoting by $\bA_i$ the $i$th column vector of $\bA$, the coherence of $\bA$ is defined as
$$
\mu(\bA) = \max_{i \neq j} \frac{| \bA_i^T \bA_j|}{\|\bA_i\| \, \|\bA_j\|}.
$$
\begin{prp} \label{prp:coherence}
Assume that $\eta < 3/4$ in \eqref{eta}.  Then
$$
\mu(\bB) \leq (1 + \eta_2) \, \mu(\bA),
$$
where
$$
\eta_2 := (2 \sqrt{1-\eta} -1)^{-2} -1.
$$
Note that $\eta_2 =  2 \eta + O(\eta^2)$, with $\eta_2 < 5 \eta$ when $\eta < 1/2$.
\end{prp}

%Concretely, when the entries of $\bA$ are i.i.d.~$N(0,1/n)$, $\eta \sim 2 \sqrt{n/p} + n/p$ (in probability), so that $\eta < 3/4$ for $n/p < 0.1$, when $n$ and $p$ are both large.

\begin{proof}
To prove the proposition we develop an upper bound on the numerator $| \bB_i^T \bB_j |$ of $\mu(\bB)$, and a lower bound on the denominator elements $\|\bB_i\|$.
For $i \neq j$, we have
\beqn
| \bB_i^T \bB_j |
&=& | \bA_i^T \bQ_1^{-1} \bA_j| \\
&\leq& | \bA_i^T \bA_j| + | \bA_i^T (\bQ_1^{-1} -\bI) \bA_j| \\
&\leq& (1 + \eta_1) \, | \bA_i^T \bA_j|,
\eeqn
by \eqref{eta1}. 

We now lower bound $\|\bB_i\|$ in terms of $\|\bA_i\|$ and $\eta$.
In parallel with the proof of \prpref{RIP}, we express $\bQ_1^{-1/2} -\bI$ as a power series
$$
\bQ_1^{-1/2} -\bI = \sum_{k \geq 1} c_k (\bI -\bQ_1)^k,
$$
where $c_k$ are the coefficients in the Taylor expansion of $(1 - x)^{-1/2}$.
Taking norms on both sides of the inequality, we obtain
\beqn
\|\bQ_1^{-1/2} - \bI\|_2
&\leq& \sum_{k \geq 1} c_k \|\bQ_1 -\bI\|_2^k \\
&\leq& \sum_{k \geq 1} c_k \eta^k = (1 - \eta)^{-1/2} -1.
\eeqn
Therefore,
\beqn
\|\bB_i\|
&=& \|\bQ_1^{-1/2}\bA_i\| \\
&\geq& \|\bA_i\| - \|(\bQ_1^{-1/2} -\bI) \bA_i\| \\
&\geq& (1 - \eta_3) \, \|\bA_i\| ,
\eeqn
where $\eta_3=(1 - \eta)^{-1/2} -1$.
All together, we have
$$
\frac{| \bB_i^T \bB_j|}{\|\bB_i\| \, \|\bB_j\|} \leq \frac{(1 + \eta_1) \, | \bA_i^T \bA_j|}{(1 - \eta_3)^2 \|\bA_i\| \, \|\bA_j\|},
$$
with $(1 + \eta_1)/(1 - \eta_3)^2  = 1 + \eta_2$ by definition of $\eta_2$.
\end{proof}

\section{Conclusion}
Though the CS literature is almost silent on the effect of pre-measurement noise on recovery performance, in this letter we made the point that it may have a substantial impact on SNR.  Indeed, we showed that, for the most common measuring schemes used in CS, the model with pre-measurement noise is, after whitening, equivalent to a standard model with only measurement noise, modulo a change in measurement matrix and an increase in the noise variance by a factor of $p/n$.  We provided rigorous bounds on the RIP constants and the coherence of the new measurement matrix which show that, as $n,p \to \infty$ with $p/n \to 0$, the constants are essentially unchanged.  As the performance of standard recovery methods are formulated in terms of the RIP constants and the coherence, this shows that, in this regime, these methods operate as usual, except with noise folding lead to a noise variance multiplied by $p/n$.

%\appendices
%\section{}

%\section*{Acknowledgment}

%The work of EAC is partially supported by the Office of Naval Research under grant N00014-09-1-0258.  And the work of YCE is supported in part by the Israel Science Foundation under grant 1081/07, and by the Ollendorf foundation.

%\ifCLASSOPTIONcaptionsoff
%  \newpage
%\fi


%\setlength{\bibsep}{0.0pt}
\bibliographystyle{abbrv}

\end{document}